\documentclass[11pt]{amsart}
\usepackage{latexsym,amssymb,amsmath,youngtab}
\usepackage{youngtab}
\usepackage{tikz}
\usepackage{amsmath,amsthm,amssymb,amscd}
\usepackage[mathscr]{eucal}
\usepackage{verbatim, hyperref}
\usepackage{color}
\newcommand*{\textlabel}[2]{%
  \edef\@currentlabel{#1}
  \phantomsection
  #1\label{#2}
}

\newtheoremstyle{custom}
  {3pt}
  {3pt}
  {\slshape}
  {}
  {\bfseries}
  {.}
  { }
   {}
\theoremstyle{custom}
\newtheorem{theorem}{Theorem}[section]
\newtheorem{proposition}[theorem]{Proposition}

\newtheorem{proposition/definition}[theorem]{Proposition/Definition}

\theoremstyle{definition}

\theoremstyle{remark}
\newtheorem{remark}[theorem]{Remark}

\newtheoremstyle{exercise}
  {3pt}
  {6pt}
  {}
  {}
  {\bfseries}
  {:}
  { }
   {}
\theoremstyle{exercise}
\newtheorem{exercise}[theorem]{Exercise}
\newtheoremstyle{exercises}
  {3pt}
  {6pt}
  {}
  {}
  {\bfseries}
  {:}
  {\newline}
   {}

\theoremstyle{exercise}
\newtheorem{exercises}[theorem]{Exercises}

\input epsf
\def\boxit#1{\vbox{\hrule height1pt\hbox{\vrule width1pt\kern3pt
  \vbox{\kern3pt#1\kern3pt}\kern3pt\vrule width1pt}\hrule height1pt}}

\def\trank{\text{rank}}

\def\BC{\mathbb C}

\def\BP{\mathbb P}
\def\pp#1{\mathbb P^{#1}}

\def\pp#1{{\mathbb P}^{#1}}
\def\tdim{{\rm dim}}

\def\hd{,...,}

\def\inv{{}^{-1}}

\def\trace{{\rm trace}}

\def\cS{{\mathcal S}}

\def\11{\mathbf 1}

\def\fsl{{\mathfrak {sl}}}

\def\l{\lambda}

\def\o{\omega}

\def\ot{{\mathord{ \otimes } }}
\def\op{{\mathord{\,\oplus }\,}}
\def\otc{{\mathord{\otimes\cdots\otimes}\;}}

\def\ra{{\mathord{\;\rightarrow\;}}}

\def\La#1{\Lambda^{#1}}

\def\frak{\mathfrak}
\def\fsl{\frak s\frak l}

\def\op{\oplus}
\def\BZ{\Bbb Z}

\def\bb#1#2#3{b^{#1}_{{#2}{#3}}}

\def\op{\oplus}

\def\l{\lambda}

\def\FS{\mathfrak  S}

\def\BP{\mathbb  P}
\def\BC{\mathbb  C}

\def\pp#1{\mathbb  P^{#1}}

\def\cc#1#2{C^{#1}_{#2}}

\def\hd{, \hdots ,}

\def\inv{{}^{-1}}

\def\La#1{\Lambda^{#1}}

\def\pp#1{\mathbb  P^{#1}}

\def\ra{\rightarrow}

\def\ttrace{\operatorname{trace}}

\def\tdim{\operatorname{dim}}

\def\trank{\operatorname{rank}}

\def\be{\begin{equation}}
\def\ene{\end{equation}}

\def\aa#1#2{a^{#1}_{#2}}
\def\bb#1#2{b^{#1}_{#2}}
\def\cc#1#2{c^{#1}_{#2}}

\def\G{\Gamma}

\newcommand{\Id}{\operatorname{Id}}

\def\Mn{M_{\langle \nnn \rangle}}\def\Mthree{M_{\langle 3\rangle}}

\def\Mtwo{M_{\langle 2\rangle}}\def\Mthree{M_{\langle 3\rangle}}

\def\aa#1#2{a^{#1}_{#2}}
\def\bb#1#2{b^{#1}_{#2}}

\def\trank{{\mathrm {rank}}}

\def\nnn{\bold n}

\begin{document}

\author{Luca Chiantini}
\address{Dipartimento di Ingegneria dell'Informazione e Scienze Matematiche, Universita' di Siena,  53100 Siena, Italy}
\email{luca.chiantini@unisi.it}

\author{Christian Ikenmeyer}
\address{Max Planck Institute for Informatics, Saarland Informatics Campus, Building E1.4, D-66123 Saarbr\"ucken, Germany
}
\email{cikenmey@mpi-inf.mpg.de}

\author{J.M. Landsberg}
\address{
Department of Mathematics\\
Texas A\&M University\\
Mailstop 3368\\
College Station, TX 77843-3368, USA}
\email{jml@math.tamu.edu}
\author{Giorgio Ottaviani}
\address{Dipartimento di Matematica e Informatica U. Dini, viale Morgagni 67/A, 50134 Firenze, Italy }
\email{ottavian@math.unifi.it}
\thanks{Landsberg supported by NSF grant DMS-1405348.}
\title[Geometry and Strassen's algorithm]{The geometry of rank decompositions of matrix multiplication I: $2\times 2$ matrices}
\begin{abstract}
This is the first in a series of papers on rank decompositions of the matrix multiplication tensor.
In this paper we:  establish  general facts about rank decompositions of tensors,
describe potential ways to search for new matrix multiplication decompositions, give a
geometric proof of the theorem of \cite{DBLP:journals/corr/Burichenko14} establishing the symmetry group
of Strassen's algorithm, and present  two  particularly nice subfamilies  in the Strassen family
of decompositions.
\end{abstract}
\maketitle

\section{Introduction}
This is the first in a planned series of papers
on the geometry of  rank decompositions of the matrix multiplication tensor
$\Mn\in \BC^{\nnn^2}\ot \BC^{\nnn^2}\ot\BC^{\nnn^2}$. Our goals for the series are
to determine possible symmetry groups for potentially optimal (or near optimal)  decompositions of the matrix
multiplication tensor and  eventually to derive new decompositions based on symmetry
assumptions.
 In this paper we study Strassen's rank $7$ decomposition of $\Mtwo$, which we denote $ \cS tr$.
In the next paper \cite{BILR}   new decompositions of $\Mthree$ are presented and   their symmetry groups are described. 
Although this project began before the papers  \cite{DBLP:journals/corr/Burichenko14,DBLP:journals/corr/Burichenko15} 
appeared, we have benefited greatly from them in our study.

We begin in \S\ref{strvo} by reviewing Strassen's algorithm as a tensor decomposition.
Then in \S\ref{symfam} we explain basic facts about rank decompositions of tensors with symmetry, in particular,
that the decompositions come in families, and each member of the family has the same abstract symmetry group.
While these abstract groups are all the same, for practical purposes (e.g., looking for new decompositions),
some realizations are more useful than others. We review the symmetries of the matrix multiplication tensor
in \S\ref{mnsymrev}. After these generalities, in \S\ref{strfam}
we revisit the Strassen family and display a particularly convenient subfamily. We examine the Strassen family
from a projective perspective in \S\ref{projper}, which renders much of its symmetry transparent. Generalities
on the projective perspective enable  a very short proof of 
the upper bound in Burichenko's determination
of the symmetries of Strassen's decomposition  \cite{DBLP:journals/corr/Burichenko14}. The projective perspective
and emphasis on symmetry also enable  two geometric proofs that Strassen's expression actually is a decomposition
of $\Mtwo$, which we explain in \S\ref{proveisstr}.

 \subsection*{Notation and conventions} $A,B,C,U,V,W$ are vector spaces, $GL(A)$ denotes the group of invertible
 linear maps $A\ra A$,   and
 $PGL(A)=GL(A)/\BC^*$ the group of projective transformations of projective space $\BP A$.
 If $a\in A$, $[a]$ denotes the corresponding point in projective space. $\FS_d$ denotes the permutation group on $d$ elements. 
 Irreducible representations of $\FS_d$ are indexed by partitions. We let $[\pi]$ denote the irreducible $\FS_d$ module
 associated to the partition $\pi$. 

\subsection*{Acknowledgements} This project began at a November 2013 \lq\lq Research in pairs\rq\rq\ program
at Mathematisches Forschungsinstitut Oberwolfach. The authors thank the institute for its hospitality and a great work environment.
{CI was at Texas A\&M University during most of the time this research was conducted.}

\section{Strassen's algorithm}\label{strvo}
 In 1968, V. Strassen set out to prove the standard  algorithm
 for multiplying $\nnn\times \nnn$ matrices was optimal in the sense that no algorithm using
  fewer multiplications exists. Since he anticipated this would be difficult
  to prove, he tried to show it just for two by two matrices. His spectacular failure opened up a whole new area of research:
  Strassen's algorithm for multiplying 
$2\times 2$ matrices $a,b$   using  seven scalar multiplications \cite{Strassen493} is as follows: 
Set
\begin{align*}
\label{starptwo}I&= (\aa 11 + \aa 22)(\bb 11 + \bb 22),\\ 
\nonumber II&=(\aa 21 + \aa 22)\bb 11, \\
\nonumber III&= \aa 11(\bb12-\bb 22)\\
\nonumber IV&=\aa 22(-\bb 11+\bb 21)\\
\nonumber V&=(\aa 11+\aa 12)\bb 22\\
\nonumber VI&= (-\aa 11+\aa 21)(\bb11 +\bb12),\\
\nonumber VII&=(\aa 12 -\aa 22)(\bb 21 + \bb 22).
\end{align*}

Set
\begin{align*}
\cc 11&= I+IV-V+VII,\\
\cc 21&= II+IV,\\
\cc 12 &= III + V,\\
\cc 22 &= I+III-II+VI.
\end{align*} 
Then $c=ab$.

To better see symmetry,    view  matrix multiplication as a trilinear map $(X,Y,Z)\mapsto \ttrace(XYZ)$
and in tensor form. To view it more invariantly,
let $U,V,W=\BC^2$, let $A=U^*\ot V$, $B=V^*\ot W$, $C=W^*\ot U$ and
consider $\Mtwo\in (V\ot U^*)\ot (W\ot V^*)\ot (U\ot W^*)$, where $\Mtwo=\Id_U\ot \Id_V\ot\Id_W$ with the factors re-ordered
(see, e.g., \cite[\S 2.5.2]{MR2865915}).
Write
\be\label{stdvects}
u_1=\begin{pmatrix} 1\\ 0\end{pmatrix}, \ u_2=\begin{pmatrix} 0\\ 1\end{pmatrix}, \ 
u^1=(1,0) \ u^2=(0,1)
\ene
and set $v_j=w_j=u_j$ and $v^j=w^j=u^j$. 
Then Strassen's algorithm becomes the following tensor decomposition
\begin{align}
\label{stra}\Mtwo=&(v_1u^1 +v_2u^2)\ot (w_1v^1 +w_2v^2)\ot(u_1w^1 +u_2w^2)\\
\label{strb}&[v_1u^1\ot w_2(v^1-v^2)\ot (u_1+u_2)w^2\\
&\nonumber 
+(v_1+v_2)u^2\ot w_1v^1\ot u_2(w^1-w^2)   \\
&\nonumber +v_2(u^1-u^2)\ot (w_1+w_2)v^2\ot u_1w^1]\\
\label{strc}&+[v_2u^2\ot w_1(v^2-v^1)\ot (u_1+u_2)w^1\\
&\nonumber +(v_1+v_2)u^1\ot w_2v^2\ot u_1(w^2-w^1)   \\
&\nonumber +v_1(u^2-u^1)\ot (w_1+w_2)v^1\ot u_2w^2].
\end{align}
 Note that this is the sum of seven rank one tensors, 
  while the standard algorithm in tensor format has   eight rank one summands. 

Introduce the notation
$$
\langle v_i  u^j \ot w_k  u^l\ot u_p  w^q\rangle_{\BZ_3}
:= v_i  u^j \ot w_k  u^l\ot u_p  w^q
+
v_k  u^l \ot w_p  u^q\ot u_i  w^j
+
v_p  u^q \ot w_i u^j\ot u_k  w^l.
$$
Then the decomposition  becomes
\begin{align}
\label{stra3}\Mtwo=&(v_1u^1 +v_2u^2)\ot (w_1v^1 +w_2v^2)\ot(u_1w^1 +u_2w^2)\\
&+\label{strb3} \langle v_1u^1\ot w_2(v^1-v^2)\ot (u_1+u_2)w^2\rangle_{\BZ_3}\\
\label{strc3}& 
-\langle v_2u^2\ot w_1(v^1-v^2)\ot (u_1+u_2)w^1\rangle_{\BZ_3}.
\end{align}

 From this presentation we immediately see there is a cyclic $\BZ_3$ symmetry by cyclically permuting the factors
$A,B,C$. The $\BZ_3$ acting on the rank one elements in the decomposition has three orbits
If we exchange  $u_1\leftrightarrow u_2$, $u^1\leftrightarrow u^2$, $v^1\leftrightarrow v^2$, etc., 
the decomposition is also preserved
by this $\BZ_2$, with orbits \eqref{stra3} and the
exchange of the  triples,  call this an {\it internal} $\BZ_2$. These symmetries are only part of the picture. 

\section{Symmetries  and families}\label{symfam}

Let $T\in (\BC^N)^{\ot k}$.
We say $T$ has {\it rank one} if $T=a_1\otc a_k$ for some $a_j\in \BC^N$. Define the {\it symmetry group} of $T$,  $G_T\subset (GL_N^{\times k})\ltimes \FS_k$ to be the subgroup  preserving $T$, where $\FS_k$ acts by permuting the factors. 
 

 For a rank decomposition   $T=\sum_{j=1}^r t_j$ with each $t_j$ of tensor rank one, define the set
$\cS :=\{t_1\hd t_r\}$, which we also call the decomposition,  and
the {\it symmetry group of the decomposition}  $\G_{\cS }:=\{ g\in G_T\mid g\cdot \cS =\cS \}$.
 Let $\G_{\cS}'=\G_{\cS} \cap (GL(A)\times GL(B)\times GL(C))$.  Let $\cS tr$ denote Strassen's 
decomposition of $\Mtwo$.

If $g\in G_T$, then  $g\cdot \cS :=\{ gt_1\hd gt_r\}$ is also a rank decomposition of $T$.
Moreover:

\begin{proposition}\label{conjprop} For $g\in G_T$,  $\G_{g\cdot \cS }= g\G_{\cS } g\inv$.
\end{proposition}
\begin{proof} Let $h\in  \G_{\cS }$, then $ ghg\inv (gt_j)=g(ht_j)\in g\cdot \cS $ so $\G_{g\cdot \cS }\subseteq  g\G_{\cS_t} g\inv$,
but the construction is symmetric in $\G_{g\cdot \cS }$ and $ \G_{\cS }  $.
\end{proof}

Similarly for a polynomial $P\in S^d\BC^N$ and a Waring decomposition $P=\ell_1^d+\cdots + \ell_r^d$ for some $\ell_j\in \BC^N$,
and $g\in G_P\subset GL_N$, the same result holds where $\cS=\{ \ell_1\hd \ell_r\}$. 

In summary,  algorithms come in $\tdim (G_T)$-dimensional families, and each member of the family has the same abstract symmetry group.

We recall the following theorem of de Groote:

\begin{theorem}\label{degrootethm}  \cite{MR0506377} The set of rank seven decompositions of $\Mtwo$ is
the orbit $G_{\Mtwo}\cdot \cS tr $.
\end{theorem}

\section{Symmetries of $\Mn$}\label{mnsymrev}
We review the symmetry group of the matrix multiplication tensor
$$
G_{\Mn}:=\{ g\in GL_{n^2}^{\times 3}\times \FS_3\mid g\cdot \Mn=\Mn\}.
$$

One may  also consider matrix multiplication as a polynomial that happens to be multi-linear,
$\Mn\in S^3(A\op B\op C)$, and consider
$$
 \tilde G_{\Mn}:=\{ g\in GL(A\op B\op C) \mid g\cdot \Mn=\Mn\}.
$$

Note that $(GL(A)\times GL(B)\times GL(C)) \times \FS_3\subset GL(A\op B\op C)$, 
so $G_{\Mn}\subseteq \tilde G_{\Mn}$.

It is clear that $PGL_{\nnn}\times PGL_{\nnn}\times PGL_{\nnn}\times \BZ_3\subset G_{\Mn}$, the $\BZ_3$
because $\trace(XYZ)=\ttrace(YZX)$,   and the $PGL_\nnn$'s appear instead of $GL_\nnn$ because if
  we rescale by $\l \Id_U$, then $U^*$ scales by $\frac 1\l$ and there is no effect on
the decomposition.
Moreover since $\ttrace(XYZ)=\ttrace(Y^TX^TZ^T)$, we have
$PGL_n^{\times 3}\ltimes D_3\subseteq G_{\Mn}$,  
where the dihedral group $D_3$ is isomorphic to $\FS_3$,  but we denote it by $D_3$ to avoid confusion
 with a second copy of $\FS_3$ that will appear. We emphasize that this $\BZ_2$ is not contained
in either the $\FS_3$ permuting the factors  or the $PGL(A)\times PGL(B)\times PGL(C)$ acting on them.
In $\tilde G_{\Mn}$ we can also rescale the three factors by non-zero complex numbers $\l,\mu,\nu$ such
that $\l\mu\nu=1$, so we have 
$(\BC^*)^{\times 2}\times PGL_n^{\times 3}\ltimes D_3\subseteq\tilde G_{\Mn}$,

We will be primarily interested in $G_{\Mn}$. The first equality in the  following proposition appeared
in \cite[Thms. 3.3,3.4]{MR0506377} and \cite[Prop. 4.7]{DBLP:journals/corr/Burichenko15}
with   ad-hoc proofs. The second assertion appeared in \cite{MR3513064}.
We reproduce the proof from \cite{MR3513064}, as it is a special case of the result there.

\begin{proposition} $G_{\Mn}=  PGL_n^{\times 3}\ltimes D_3 $
and $\tilde G_{\Mn}=(\BC^*)^{\times 2}\times PGL_n^{\times 3}\ltimes D_3 $.
\end{proposition}
\begin{proof} It will be sufficient to show the second equality because the $(\BC^*)^{\times 2}$ acts trivially on $A\ot B\ot C$.
For polynomials, we use the method of \cite[Prop. 2.2]{MR3217699} adapted to  reducible representations.
A straight-forward Lie algebra calculation shows 
the connected component of the identity of $\tilde G_{\Mn}$ is $\tilde G_{\Mn}^0=(\BC^*)^{\times 2}\times PGL_n^{\times 3}$.
As was observed in \cite{MR3217699} the full stabilizer group must be contained in its normalizer $N(\tilde G_{\Mn}^0)$.
But the normalizer is the automorphism group of the marked Dynkin diagram for
$A\op B\op C$, which in our case is

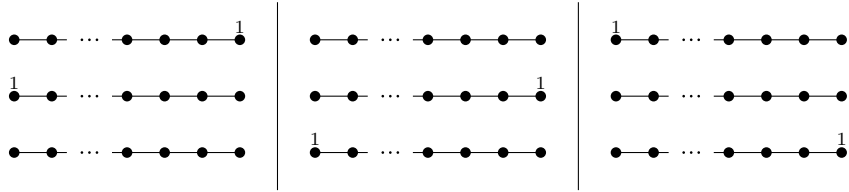
\begin{center}
\begin{tikzpicture}[scale = 0.5]
\fill (0+-1,0) circle[radius=4pt];
\fill (0+0,0) circle[radius=4pt];
\fill (0+1-0.2,0) circle[radius=1pt];
\fill (0+1,0) circle[radius=1pt];
\fill (0+1+0.2,0) circle[radius=1pt];
\fill (0+2,0) circle[radius=4pt];
\fill (0+3,0) circle[radius=4pt];
\fill (0+4,0) circle[radius=4pt];
\fill (0+5,0) circle[radius=4pt];
\draw (0+-1,0) -- (0+0.4,0);
\draw (0+2-0.4,0) -- (0+5,0);
\fill (0+-1,1.5) circle[radius=4pt];
\fill (0+0,1.5) circle[radius=4pt];
\fill (0+1-0.2,1.5) circle[radius=1pt];
\fill (0+1,1.5) circle[radius=1pt];
\fill (0+1+0.2,1.5) circle[radius=1pt];
\fill (0+2,1.5) circle[radius=4pt];
\fill (0+3,1.5) circle[radius=4pt];
\fill (0+4,1.5) circle[radius=4pt];
\fill (0+5,1.5) circle[radius=4pt];
\draw (0+-1,1.5) -- (0+0.4,1.5);
\draw (0+2-0.4,1.5) -- (0+5,1.5);
\fill (0+-1,3) circle[radius=4pt];
\fill (0+0,3) circle[radius=4pt];
\fill (0+1-0.2,3) circle[radius=1pt];
\fill (0+1,3) circle[radius=1pt];
\fill (0+1+0.2,3) circle[radius=1pt];
\fill (0+2,3) circle[radius=4pt];
\fill (0+3,3) circle[radius=4pt];
\fill (0+4,3) circle[radius=4pt];
\fill (0+5,3) circle[radius=4pt];
\draw (0+-1,3) -- (0+0.4,3);
\draw (0+2-0.4,3) -- (0+5,3);
\draw (6,-1) -- (6,4);
\fill (8+-1,0) circle[radius=4pt];
\fill (8+0,0) circle[radius=4pt];
\fill (8+1-0.2,0) circle[radius=1pt];
\fill (8+1,0) circle[radius=1pt];
\fill (8+1+0.2,0) circle[radius=1pt];
\fill (8+2,0) circle[radius=4pt];
\fill (8+3,0) circle[radius=4pt];
\fill (8+4,0) circle[radius=4pt];
\fill (8+5,0) circle[radius=4pt];
\draw (8+-1,0) -- (8+0.4,0);
\draw (8+2-0.4,0) -- (8+5,0);
\fill (8+-1,1.5) circle[radius=4pt];
\fill (8+0,1.5) circle[radius=4pt];
\fill (8+1-0.2,1.5) circle[radius=1pt];
\fill (8+1,1.5) circle[radius=1pt];
\fill (8+1+0.2,1.5) circle[radius=1pt];
\fill (8+2,1.5) circle[radius=4pt];
\fill (8+3,1.5) circle[radius=4pt];
\fill (8+4,1.5) circle[radius=4pt];
\fill (8+5,1.5) circle[radius=4pt];
\draw (8+-1,1.5) -- (8+0.4,1.5);
\draw (8+2-0.4,1.5) -- (8+5,1.5);
\fill (8+-1,3) circle[radius=4pt];
\fill (8+0,3) circle[radius=4pt];
\fill (8+1-0.2,3) circle[radius=1pt];
\fill (8+1,3) circle[radius=1pt];
\fill (8+1+0.2,3) circle[radius=1pt];
\fill (8+2,3) circle[radius=4pt];
\fill (8+3,3) circle[radius=4pt];
\fill (8+4,3) circle[radius=4pt];
\fill (8+5,3) circle[radius=4pt];
\draw (8+-1,3) -- (8+0.4,3);
\draw (8+2-0.4,3) -- (8+5,3);
\draw (14,-1) -- (14,4);
\fill (16+-1,0) circle[radius=4pt];
\fill (16+0,0) circle[radius=4pt];
\fill (16+1-0.2,0) circle[radius=1pt];
\fill (16+1,0) circle[radius=1pt];
\fill (16+1+0.2,0) circle[radius=1pt];
\fill (16+2,0) circle[radius=4pt];
\fill (16+3,0) circle[radius=4pt];
\fill (16+4,0) circle[radius=4pt];
\fill (16+5,0) circle[radius=4pt];
\draw (16+-1,0) -- (16+0.4,0);
\draw (16+2-0.4,0) -- (16+5,0);
\fill (16+-1,1.5) circle[radius=4pt];
\fill (16+0,1.5) circle[radius=4pt];
\fill (16+1-0.2,1.5) circle[radius=1pt];
\fill (16+1,1.5) circle[radius=1pt];
\fill (16+1+0.2,1.5) circle[radius=1pt];
\fill (16+2,1.5) circle[radius=4pt];
\fill (16+3,1.5) circle[radius=4pt];
\fill (16+4,1.5) circle[radius=4pt];
\fill (16+5,1.5) circle[radius=4pt];
\draw (16+-1,1.5) -- (16+0.4,1.5);
\draw (16+2-0.4,1.5) -- (16+5,1.5);
\fill (16+-1,3) circle[radius=4pt];
\fill (16+0,3) circle[radius=4pt];
\fill (16+1-0.2,3) circle[radius=1pt];
\fill (16+1,3) circle[radius=1pt];
\fill (16+1+0.2,3) circle[radius=1pt];
\fill (16+2,3) circle[radius=4pt];
\fill (16+3,3) circle[radius=4pt];
\fill (16+4,3) circle[radius=4pt];
\fill (16+5,3) circle[radius=4pt];
\draw (16+-1,3) -- (16+0.4,3);
\draw (16+2-0.4,3) -- (16+5,3);
\node at (-1,1.5+0.35) {\tiny 1};
\node at (5,3+0.35) {\tiny 1};
\node at (7,0+0.35) {\tiny 1};
\node at (13,1.5+0.35) {\tiny 1};
\node at (15,3+0.35) {\tiny 1};
\node at (21,0+0.35) {\tiny 1};
\end{tikzpicture}
\end{center}

There are three triples of marked diagrams. Call each column 
consisting of 3 marked diagrams a group.  The automorphism group of the picture is $D_3=\BZ_2\ltimes \BZ_3$,
where the $\BZ_2$ may be seen as  flipping each diagram, exchanging the first and third diagram in each group, and exchanging
the first and second group. The $\BZ_3$ comes from cyclically permuting each group and the diagrams within each group.
\end{proof}

Regarding the symmetries discussed in \S\ref{strvo},    the $\BZ_3$ is in 
the $\FS_3$ in $PGL_2^{\times 3}\times \FS_3$  and the  internal  $\BZ_2$ is in $\G_{\cS tr}'\subset PGL_2^{\times 3}$.

Thus if  $\cS $ is (the set of points of)  a  rank decomposition of $\Mn$, then
$\G_{\cS}\subset [(GL(U)\times GL(V)\times GL(W))\ltimes \BZ_3]\ltimes \BZ_2$.

We call a $\BZ_3\subset \G_{\cS}$ a {\it standard cyclic symmetry} if it corresponds to
  $(\Id,\Id,\Id)\cdot \BZ_3\subset (GL(U)\times GL(V)\times GL(W))\ltimes \BZ_3$.
  
We call a $\BZ_2\subset \G_{\cS}$ a {\it convenient transpose symmetry} if
it corresponds to  the symmetry
of $\Mn$ given by $a\ot b\ot c\mapsto a^T\ot c^T\ot b^T$.
The convenient transpose symmetry lies in
$(GL(A)\times GL(B)\times GL(C))\times \FS_2\subset
(GL(A)\times GL(B)\times GL(C))\times \FS_3$, where the  component of the transpose in $\FS_2$ switches the
last two factors and the component in $GL(A)\times GL(B)\times GL(C)$ sends each matrix to its transpose. 

\begin{remark}Since $\Mn\in (U^*\ot U)^{\ot 3}$ one could consider the larger
symmetry group considering $\Mn\in U^{\ot 3}\ot U^{*\ot 3}$ as is done in \cite{DBLP:journals/corr/Burichenko14}.
\end{remark}

 \section{The Strassen family}\label{strfam}

Since $PGL_2^{\times 3}\subset G_{\Mtwo}$,   we can replace
$u_1,u_2$ by any basis of $U$ in Strassen's decomposition, and similarly for $v_1,v_2$ and $w_1,w_2$.
In particular, 
we need not have $u_1=v_1$ etc... When we do that, the symmetries become conjugated by our change of basis
matrices. If we only use elements of the diagonal $PGL_2$ in  $PGL_2^{\times 3}$, the $\BZ_3$-symmetry remains
standard. More  subtly, the  $\BZ_3$-symmetry remains the standard
cyclic permutation of factors if we apply
  elements of $\BZ_3$ in any of the $PGL_2$'s, i.e., setting $\o=e^{\frac{2\pi i}{3}}$, 
  we can apply any of
$$
\rho(\o)= \begin{pmatrix} 0&-1\\1&-1\end{pmatrix} \ {\rm and} \ 
\rho(\o^2)= \begin{pmatrix} -1& 1\\-1&0\end{pmatrix} 
$$
to $U,V$ or $W$. 

For example, if we apply the change of basis matrices
$$
g_U=\begin{pmatrix}1&-1\\ 0& 1\end{pmatrix}\in  GL(U), \ 
g_V=\begin{pmatrix}-1&0\\ -1&1\end{pmatrix}\in  GL(V), \ 
g_W=\begin{pmatrix}0&1\\ 1&0\end{pmatrix}\in  GL(W),
$$
and take the image vectors as our new basis vectors, 
then setting $u_3=-(u_1+u_2)$ and $u^3=u^1-u^2$ and similarly for  the $v$'s and $w$'s,  the decomposition becomes:
\begin{align}
\label{strax}\Mtwo=&-( v_1u^2+v_2u^3)\ot ( w_1v^2+w_2v^3)\ot ( u_1w^2+u_2w^3) 
\\
\label{strbx}&+v_1u^1\ot w_1v^1\ot u_1w^1\\
&  \label{strbx2}
+v_3u^2 \ot w_3v^2 \ot u_3w^2   \\
&\label{strbx3} +v_2u^3 \ot w_2v^3 \ot u_2w^3 \\
\label{strcx}&- \langle v_1u^2\ot w_2v^1\ot u_3w^3\rangle_{\BZ_3}.
\end{align}

\begin{remark} The matrices in \eqref{strcx} are all nilpotent, and none of the other matrices
appearing in this decomposition are.
\end{remark}

Notice that for the first term 
$ v_1u^2+v_2u^3=   v_2u^1+v_3 u^1=v_3u^2+v_2u^1$.
Here there is a standard $\BZ_3\subset \FS_3$. There are four fixed points for this
standard $\BZ_3$: \eqref{strax},\eqref{strbx}\eqref{strbx2},\eqref{strbx3}. (In any element of the Strassen family there will be some $\BZ_3$ with
four fixed points, but the $\BZ_3$ need not be standard.)
There is also a standard $\BZ_3\subset PGL_2^{\times 3}$ embedded diagonally, that sends
$u_1\ra u_2\ra u_3$, and acting by the inverse matrix on the dual basis, and similarly for  the $v$'s and $w$'s. Under this action \eqref{strax} is fixed and
we have the cyclic permutation \eqref{strbx}$\ra$\eqref{strbx3}$\ra$\eqref{strbx2}.

If we take the standard vectors of  \eqref{stdvects} in each factor  we get
$$\Mtwo=
\begin{pmatrix} 0 &-1\\ 1&-1\end{pmatrix}^{\ot 3}
+ 
\begin{pmatrix} 1 &0\\ 0&0\end{pmatrix}^{\ot 3}
+ 
\begin{pmatrix} 0 &1\\ 0&1\end{pmatrix}^{\ot 3}
+ 
\begin{pmatrix} 0 &0\\ -1&1\end{pmatrix}^{\ot 3}
+
\langle   
\begin{pmatrix} 0 &1\\ 0&0\end{pmatrix} \ot   
\begin{pmatrix} 0 &0\\ 1&0\end{pmatrix}\ot  
\begin{pmatrix}  1 &-1\\ 1&-1\end{pmatrix}\rangle_{\BZ_3}
$$

If we want to see the $\BZ_3\subset PGL_2^{\times 3}$ more transparently, it is better to diagonalize
the $\BZ_3$ action so the first matrix becomes
$$
a=\begin{pmatrix} \o &0\\ 0 & \o^2\end{pmatrix}.
$$
where $\o:=\exp(\frac{-2\pi i}{3})$. Then for $\iota := i/\sqrt3$, $\sigma := \exp(\frac{2\pi i}{12})/\sqrt3$
we get
$$\Mtwo=
a^{\otimes 3}
+
b^{\otimes 3}
+
(\varrho(b))^{\otimes 3}
+
(\varrho^2(b))^{\otimes 3}
+
\langle
c
\otimes
\varrho(c)
\otimes
\varrho^2(c)
\rangle_{\BZ_3},
$$
where 
\[
b := \begin{pmatrix}
 \sigma & \bar\iota \\
 \iota & \bar\sigma
\end{pmatrix}, \quad
c := \begin{pmatrix}
 \iota & \iota \\
 \bar\iota & \bar\iota
\end{pmatrix}, \quad
\varrho : \BC^{2 \times 2} \to \BC^{2 \times 2}, \ \varrho(X) = a X a^{-1}.
\]
Note that $a+b+c=0$.

\section{Projective perspective}\label{projper}
Although the above description of the Strassen family of decompositions for $\Mtwo$  is satisfying, 
it becomes even more transparent with a projective perspective.

\subsection{$\Mtwo$ viewed projectively}\label{m2subsect}
Recall that $PGL_2$ acts simply transitively on the set of triples of distinct points of $\pp 1$.
So to fix a decomposition in the family, we select a triple of points in each space. We focus on $\BP U$.
Call the points $[u_1],[u_2],[u_3]$.
Then these determine three points in $\BP U^*$, $[u^{1\perp}],[u^{2\perp}],[u^{3\perp}]$.
 We choose representatives $u_1,u_2,u_3$
satisfying $u_1+u_2+u_3=0$.  We could have taken any linear relation, it just would introduce
coefficients in the decomposition. We take the most symmetric relation to keep all three points on
an equal footing. Similarly, we fix the scales on the $u^{j\perp}$ by requiring
$u^{j\perp}(u_{j-1})=1$ and $u^{j\perp}(u_{j+1})=-1$, where indices are considered mod $\BZ_3$, so   
$u_{3+1}=u_1$ and $u_{1-1}=u_3$.

In comparison with what we had before, letting the old indices be hatted, we have
$\hat u_1=u_1$, $\hat u_2=u_2$, $\hat u_3=-u_3$ and 
$\hat u^1=u^{2\perp}$, $\hat u^2=-u^{1\perp}$, and $\hat u^3=-u^{3\perp}$.
The effect is to make the symmetries of the decomposition more transparent.
Our identifications of the ordered triples $\{ u_1,u_2,u_3\}$ and $\{ v_1,v_2,v_3\}$ exactly
determine a linear isomorphism $a_0: U\ra V$, and similarly for the other pairs of vector
spaces. Note that $a_0=v_j\ot u^{j+1\perp} + v_{j+1}\ot u^{j+2\perp}$ for any $j=1,2,3$. 
 
Then
\begin{align}
\label{projm2}\Mtwo&=a_0\ot b_0\ot c_0\\
\nonumber &+\langle ( v_2\ot u^{1\perp})\ot ( w_1\ot v^{3\perp})\ot ( u_3\ot w^{2\perp})\rangle_{\BZ_3}\\
\nonumber &  + \langle ( v_3\ot u^{1\perp})\ot ( w_1\ot v^{2\perp})\ot ( u_2\ot w^{3\perp})\rangle_{\BZ_3}.
\end{align}
With this presentation, the $\FS_3\subset PGL_2\subset PGL_2^{\times 3}$ acting by permuting the indices
transparently preserves the decomposition, with two orbits, the fixed point $a_0\ot b_0\ot c_0$
and the orbit of $( v_2\ot u^{1\perp})\ot ( w_1\ot v^{3\perp})\ot ( u_3\ot w^{2\perp})$.

\begin{remark} Note that here there are no nilpotent matrices appearing. 
\end{remark}

\begin{remark}
The geometric picture of the decomposition of $\Mtwo$ can be rephrased
as follows. Consider the space of linear isomorphisms $U\to V$ (mod
scalar multiplication) as the
projective space $\mathbb P^3$ of $2\times 2$ matrices, in which we fix
coordinates, coming from the choice of basis
for $U,V$. The choice of basis also determines an identification between
$U$ and $V$. Then $a_0$ represents in $\mathbb
P^3$ a point of rank $2$, which can be taken as the identity in the
choice of coordinates. The other $6$ points
$Q_i=u_i\otimes u^{j\perp}$ appearing in the first factor of the
decomposition can be determined as follows. The
points $P_i=u_i\otimes u^{i\perp}$ (in the identification) represent the
choice of $3$ points in the conic obtained
by cutting with a plane (e.g. the plane of traceless matrices) the
quadric  $q=Seg(\pp 1\times \pp 1)$  of matrices of rank $1$. Through
each $P_i$ one finds  lines of the two rulings of $q$, call then
$L_i,M_i$. Then the six points $Q_i$ are given
by:
$$ Q_1=L_1\cap M_2, \ Q_2=L_2\cap M_3, \ Q_3=L_3\cap M_1$$
$$ Q_4=M_1\cap L_2,\ Q_5=M_2\cap L_3,\ Q_6=M_3\cap L_1.$$
An analogue of the construction determines the seven points in the other
two factors of the tensor product, so that the $7$ final summands can be
determined
combinatorially and the $\mathbb Z_2, \mathbb Z_3$ symmetries can be
easily recognized.

The geometric construction can be generalized to higher dimensional
spaces, so it could insight  for extensions to larger matrix multiplication
tensors. The difficult part is to determine 
how one should combine the points
constructed in each factor of the tensor product, in order to produce a
decomposition of $M_{\langle \nnn\rangle}$.
\end{remark}

When we view \eqref{strax} projectively, we get

 \begin{align}
\label{straxx}\Mtwo=& ( v_1u^{1\perp}+v_2u^{3\perp})\ot ( w_1v^{1\perp}+w_2v^{3\perp})\ot ( u_1w^{1\perp}+u_2w^{3\perp}) 
\\
\label{strbxx}&+v_1u^{2\perp}\ot w_1v^{2\perp}\ot u_1w^{2\perp}\\
&  \label{strbxx2}
+v_3u^{1\perp} \ot w_3v^{1\perp} \ot u_3w^{1\perp}   \\
&\label{strbxx3} +v_2u^{3\perp} \ot w_2v^{3\perp} \ot u_2w^{3\perp} \\
\label{strcxx}&  \langle v_1u^{1\perp}\ot w_2v^{2\perp}\ot u_3w^{3\perp}\rangle_{\BZ_3}.
\end{align}

With this presentation,   $\FS_3\subset \G_{\cS}'$ is again transparent.

\subsection{Symmetries of $\G_{\cS tr}$}
Let $\Mn=\sum_{j=1}^r t_j$ be a rank decomposition for $\Mn$ and write $t_j=a_j\ot b_j\ot c_j$.
Let $\bold r_j =(r_{A,j} ,r_{B,j},r_{C,j}):=(\trank(a_j), \trank(b_j),\trank(c_j))$, and
let $\tilde {\bold r}_j$ denote the unordered triple.
The following   proposition is  clear:

\begin{proposition}
Let $\cS$ be a rank decomposition of $\Mn$. Partition $\cS$ by un-ordered rank triples into
disjoint subsets: $\cS:=\{ \cS_{1,1,1}, \cS_{1,1,2}\hd \cS_{n,n,n}\}$.
Then 
$\G_{\cS}'$ preserves each $\cS_{s,t,u}$.  
\end{proposition}

We can say more about rank one elements:

If $a\in U^*\ot V$ and  $\trank(a)=1$, then there are unique points $[\mu]\in \BP U^*$ and
$[v]\in \BP V$ such that $[a]=[\mu\ot v]$.

Now   given a decomposition  $\cS$ of $\Mn$,
define $\cS_{U^*}\subset \BP U^*$ and $\cS_U\subset \BP U$ to correspond to the
elements appearing in $\cS_{1,1,1}$. 
Then $\G_{\cS}'$ preserves $\cS_U$ and $\cS_{U^*}$.

In the case of Strassen's decomposition $\cS tr_U$ is a configuration of three points in $\pp 1$,
so {\it a priori} we must have $\G_{\cS tr}'\cap PGL(U)\subset \FS_3$. If we insist on
the standard $\BZ_3$-symmetry (i.e., restrict to the subfamily of decompositions where there is a 
standard cyclic
symmetry), there is just one $PGL_2$ and we have $\G_{\cS tr}'\subseteq \FS_3$. Recall that this
is no loss of generality as the full symmetry group is the same for all decompositions 
in the family.
We conclude $\G_{\cS tr}\subseteq \FS_3\times D_3$. 
We have already seen $\FS_3\times \BZ_3\subset \G_{\cS tr}$, Burichenko \cite{DBLP:journals/corr/Burichenko14} shows
that in addition there is a non-convenient $\BZ_2$ obtained by taking the convenient $\BZ_2$
(which sends the decomposition to another decomposition in the family) and then conjugating
by  $\begin{pmatrix}  0 & -1\\ 1 & 0\end{pmatrix}\subset PGL_2\subset PGL_2^{\times 3}$ which sends the
decomposition back to $\cS tr$. We recover (with a new proof of the upper bound) Burichenko's theorem:

\begin{theorem} \cite{DBLP:journals/corr/Burichenko14} The symmetry group of
Strassen's decomposition of $\Mtwo$ is
$\FS_3\times D_3\subset PGL_2^{\times 3}\times D_3=G_{\Mtwo}$.
\end{theorem}

\section{How to prove Strassen's decomposition is actually matrix multiplication}\label{proveisstr}
The group $\G_{\cS tr}$ acts on $(U^*\ot U)^{\ot 3}$ (in different ways, depending on the choice of decomposition in the
family). Say we did not know $\cS tr$ but did know its symmetry group. Then we could look for it inside the space
of $\G_{\cS tr}$ invariant tensors. In future work we plan to take candidate symmetry groups for matrix multiplication
decompositions and look for decompositions with elements from these subspaces. In this paper we simply illustrate
the idea by going in the other direction: furnishing a proof that $\cS tr$ is a decomposition of $\Mtwo$, by
using the invariants to reduce the computation to a simple verification. We accomplish this  in \S\ref{invarten} below.
We first give yet another proof that Strassen's decomposition is matrix multiplication using the fact
that $\Mtwo$ is characterized by its symmetries.

\subsection{Proof that Strassen's algorithm works via characterization by  symmetries}
Here is  a   
proof that illustrates another potentially useful property of $\Mn$:  it  is {\it characterized by its symmetry group} \cite{MR3513064} 
Any $T\in (U^*\ot V)\ot (V^*\ot W)\ot (W^*\ot U)$ that is invariant under $PGL(U)\times PGL(V)\times PGL(W)\ltimes D_3$ is 
 up to scale to $\Mn$. 
 Any $T\in (U^*\ot V)\ot (V^*\ot W)\ot (W^*\ot U)$ that is invariant under
 a group isomorphic to $PGL(U)\times PGL(V)\times PGL(W)\ltimes D_3$ is 
  $GL(A)\times GL(B)\times GL(C)\times \FS_3$-equivalent  up to scale to $\Mn$. 

\begin{remark} $\Mn$ is also characterized as a polynomial by its symmetry group $\tilde G_{\Mn}$, and
any $T\in (U^*\ot V)\ot (V^*\ot W)\ot (W^*\ot U)$ that is invariant under $PGL(U)\times PGL(V)\times PGL(W)$ is 
 up to scale to $\Mn$.   However, it is not
characterized up to  $GL(A)\times GL(B)\times GL(C)$-equivalence by $G_{\Mn}'$ in the strong sense of up to isomorphism because 
$(X,Y,Z)\mapsto \ttrace(YXZ)$ has an isomorphic symmetry group but is not $GL(A)\times GL(B)\times GL(C)$-equivalent.
\end{remark}

By the above discussion, we  only need to check the right hand side of \eqref{projm2}
is invariant under $PGL(U)\times PGL(V)\times PGL(W)$ and to check its scale. But by symmetry, it is sufficient to check it
is invariant under $PGL(U)$. For this it is sufficient to check it is annihilated by $\fsl(U)$, and again
by symmetry, it is sufficient to check it is annihilated by $ u_1\ot u^{1\perp} $,
which is a simple calculation.

\subsection{Spaces of invariant tensors}\label{invarten}

As an $\FS_3$-module $ A= U^*\ot V=[21]\ot [21]=[3]\op [21]\op [1^3]$. 
In what follows we use the decompositions:
\begin{align*}
S^2[21]&=[3]\op [21]\\
\La 2[21]&=[1^3]\\
S^3[21]&=[3]\op[21]\op[1^3].
\end{align*}

The space of standard cyclic $\BZ_3$-invariant tensors in $A^{\ot 3}=S^3A\op S_{21}A^{\op 2}\op \La 3 A$ is $S^3A\op \La 3 A$.
Inside the space
 of  $\BZ_3$-invariant vectors we want to find instances of the trivial $\FS_3$-module $[3]$  in 
$S^3([3]\op [2,1]\op [1^3])\op \La 3([3]\op [2,1]\op [1^3])$.
We have
\begin{align*}
S^3([3]\op [2,1]\op [1^3])
=&S^3[3]\op S^2[3]\ot [2,1]
\op S^2[3]\ot[1^3]\op [3]\ot S^2[2,1]
\op [3]\ot [21]\ot [1^3]\\
& \op [3]\ot S^2[1^3]\ot S^3[21]\op S^2[21]\ot[13]\op[21]\ot S^2[1^3]\op S^3[1^3]
\end{align*}
and four factors contain (or are) a trivial representation: $S^3[3],[3]\ot S^2[2,1],  [3]\ot S^2[1^3], S^3[21]$
Similarly
$$
\La 3([3]\op [21]\op [1^3])=
\La 2[21]\ot [3]\op \La 2[21]\ot [1^3]\op [3]\ot [21]\ot[1^3]
$$
of which $\La 2[21]\ot [1^3]$ is the unique trivial submodule.

In summary:

\begin{proposition} The space of $\FS_3\times \BZ_3$ invariants in $(U^*\ot U)^{\ot 3}$ when $\tdim U=2$ is
five dimensional.
\end{proposition}

By a further direct calculation  we obtain:

\begin{proposition} The space of $\FS_3\times D_3$ invariants in $(U^*\ot U)^{\ot 3}$ when $\tdim U=2$ is
four dimensional.
\end{proposition}

So if we knew there were an $\FS_3\times D_3$ invariant decomposition of $\Mtwo$, it would be a simple
calculation to find it as a linear combination of four basis vectors of the $\FS_3\times D_3$-invariant tensors.
In future work we plan to assume similar invariance for larger matrix multiplication tensors to shrink the search
space to manageable size.

\bibliographystyle{amsplain}
 
\bibliography{Lmatrix}

\end{document}